\documentclass[letterpaper, 10 pt, conference]{ieeeconf}

\IEEEoverridecommandlockouts

\usepackage{amssymb,amsmath}

\usepackage{amsthm}
\usepackage{mathtools}
\usepackage{multirow}
\usepackage{graphicx}
\usepackage{comment}
\usepackage[sort,compress,noadjust]{cite}
\usepackage[font=footnotesize]{subcaption}
\usepackage[font=footnotesize]{caption}

\usepackage{amsfonts}
 \usepackage{tabularx}
\usepackage{graphicx}
\usepackage{array}
\usepackage{float}
\usepackage{hhline}
\usepackage{algorithmic}
\usepackage[linesnumbered, ruled, vlined]{algorithm2e}
\usepackage[colorlinks,urlcolor=blue]{hyperref}
\usepackage{setspace}
\usepackage{colortbl}
\usepackage{arydshln}
\SetCommentSty{mycommfont}

\usepackage{accents}
\newcommand{\ubar}[1]{\underaccent{\bar}{#1}}

\usepackage{placeins}

\theoremstyle{plain}
\newtheorem{theorem}{Theorem}
\theoremstyle{plain}
\newtheorem{proposition}{Proposition}
\theoremstyle{plain}

\theoremstyle{plain}

\theoremstyle{definition}
\newtheorem{assumption}{Assumption}
\theoremstyle{definition}
\newtheorem{definition}{Definition}
\theoremstyle{remark}
\newtheorem{remark}{Remark}

\usepackage[capitalize]{cleveref}
\crefformat{equation}{(#2#1#3)}
\Crefformat{equation}{Equation~(#2#1#3)}
\Crefname{equation}{Equation}{Eqs.}

\title{\LARGE\bf Adaptive Robust Control Contraction Metrics: \\ Transient Bounds in Adaptive Control with Unmatched Uncertainties}

\author{Samuel G. Gessow and Brett T. Lopez%
\thanks{Both authors are with the VECTR Laboratory, University of California, Los Angeles, Los Angeles CA, {\tt\footnotesize \{sgessow,btlopez\}@ucla.edu}}
}

\begin{document}
\maketitle
\thispagestyle{empty}
\pagestyle{empty}


\begin{abstract}
This work presents a new sufficient condition for synthesizing nonlinear controllers that yield bounded closed-loop tracking error transients despite the presence of unmatched uncertainties that are concurrently being learned online.
The approach utilizes contraction theory and addresses fundamental limitations of existing approaches by allowing the contraction metric to depend on the unknown model parameters.
This allows the controller to incorporate new model estimates generated online without sacrificing its strong convergence and bounded transients guarantees. 
The approach is specifically designed for trajectory tracking so the approach is more broadly applicable to  adaptive model predictive control as well.
Simulation results on a nonlinear system with unmatched uncertainties demonstrates the approach.

\end{abstract}

\section{Introduction}
\label{sec:introduction}

There is a growing need for nonlinear controllers that guarantee the closed-loop tracking error has bounded transients even when parametric model uncertainty is present.
For example, trajectory planning or high-level decision-making can be improved by leveraging transient bounds to predict how the closed-loop tracking error will evolve in time, leading to less myopic behaviors.
Furthermore, in order to limit conservatism, these controllers must be able to incorporate online model estimates without losing their strong performance guarantees.
However, concurrent control and model estimation, i.e., indirect adaptive control, that preserves transient bounds is nontrivial as the controller must simultaneously account for model estimation error and the transients associated with model estimation.
Traditionally, trajectory tracking with model error is addressed by establishing the input-to-state convergent (ISC) \cite{pavlov2006uniform} or (stronger) incremental input-to-state stable ($\delta$ISS) \cite{angeli2002lyapunov} property.
If the model is also being estimated online then the controller must be robust to not only the uncertainty in the model but also the transients of the parameter estimates \cite{krstic1995nonlinear,zamani2013backstepping}.
These conditions are pragmatically difficult to establish though for general uncertain nonlinear systems.
We will show that being ISC with respect to estimation transients can be bypassed and instead replaced with a simpler condition that eases the design of robust high-performance controllers without scarifying tracking error transient bounds.

A sufficient condition for bounded tracking error transients is the existence of an input-to-state convergent control Lyapunov function (ISC-clf) where the closed-loop system is ISC with respect to the unknown model parameters.
For the case where the model uncertainties are unmatched (outside the span of the control input matrix), the ISC-clf must be model-dependent, i.e., a function of the unknown parameters.
In essence one must find a \emph{family} of ISC-clf's. 
An immediate consequence of this dependency is that the ISC-clf itself must change as the model is updated.
This introduces significant complexity in the controller design since the transient bounds will only hold if the controller either anticipates how the model is changing or is ISC with respect to the estimation transients.
Since predicting the model estimation transients is unrealistic, imposing ISC with respect to the estimation transients is the primary approach used in the literature.
Aside from systems in strict-feedback form, systematically computing an ISC-clf for general uncertain nonlinear systems is difficult, and typically requires assumptions about the form of the ISC-clf, i.e., restricting it to take a quadratic form.
Conversely, contraction and contraction metrics \cite{lohmiller1998contraction,manchester2017control,singh2017robust,manchester2018robust} is a promising new framework to systematically compute feedback policies for stabilizable nonlinear systems.
A key advantage of contraction is that the metric search can be formulated as a convex optimization and is therefore applicable to a much broader class of systems. 
While robust \cite{singh2017robust,manchester2018robust,zhao2022tube} and direct adaptive \cite{lopez2020adaptive,lopez2021universal} control techniques based on contraction metrics have been proposed, an adaptive robust control strategy which guarantees 1) bounded tracking error transients and 2) convergence to the desired trajectory as the model estimate is improved has yet to be developed.

The main contribution of this work is a new contraction-based adaptive robust control framework that possesses desirable convergence and bounded transients guarantees for systems with unmatched uncertainties being estimated online.
Specifically, we introduce the \emph{adaptive robust control contraction metric} which, when used in feedback, can incorporate new estimates of the uncertain model parameters to improve closed-loop performance without sacrificing its strong guarantees.
Central to the approach is allowing the metric to depend on the estimated model parameters.
While this generally presents unique challenges in terms of ensuring the tracking error transients remain bounded as the model is changing online, we present a new sufficient condition for the metric that preserves transient bounds and can be immediately added to existing algorithms that numerically construct contraction metrics.
Trajectory tracking results for a system not feedback linearizable nor in strict-feedback form with unmatched uncertainties confirm the theoretical properties and showcases the advantage of the approach when compared to existing robust techniques.

\textit{Notation:} Symmetric positive definite $n \times n$ matrices are denoted $\mathcal{S}_{+}^n$. Positive scalars are denoted  $\mathbb{R}_+$ and strictly positive scalars are $\mathbb{R}_{>0}$. 
The infinity norm of a signal $x(t)$ is $||x(t)||_\infty$ and the $L_1$ norm at time $t$ is $|x(t)|$. 
The shorthand notation of a function $T$ parameterized by a vector $a$ with vector argument $z$ is $T_a(z) \triangleq T(z;a)$.
The directional derivative of a smooth matrix $M : \mathbb{R}^n \rightarrow \mathcal{S}^n_+$ along a vector field $v : \mathbb{R}^n \rightarrow \mathbb{R}^n$ is $\partial^\circ_v M(x) = \sum_i^n \nabla_{x_i} M(x) \, \dot{x}_i$ or $\partial_v M(x) = \sum_i^n \nabla_{v_i} M(x) \, \dot{v}_i$.

\section{Problem Formulation and Preliminaries}
\label{sec:problem}
Consider the uncertain control affine nonlinear system
\begin{equation}
\label{eq:system}
    \dot x = f(x) + \Delta(x)^\top\theta + B(x)u,
\end{equation}
with state $x \in \mathbb{R}^n$, control input $u \in \mathbb{R}^m$, nominal dynamics $f: \mathbb{R}^n \rightarrow \mathbb{R}^n$, and control input matrix $B: \mathbb{R}^n \rightarrow \mathbb{R}^{n \times m}$ with columns $b_i(x)$ for $i=1,\dots,m$. 
The model parameters $\theta \in \mathbb{R}^p$ are unknown but with known dynamics $\Delta: \mathbb{R}^n \rightarrow \mathbb{R}^{p \times n}$ with rows $\Delta_i(x)$ for $i=1,\dots,p$. 
No restrictions are placed on the model uncertainties for generality, hence they are allowed to be unmatched, i.e., $\Delta(x)^\top \theta \notin \mathrm{span}\{B(x)\}$.
Shorthand notation for \cref{eq:system} will be $\dot{x} = F_{\theta}(x,u)$.
We make the following assumptions about the uncertain parameters $\theta$ and their estimates $\hat{\theta}$.
\begin{assumption}
    \label{assumption:theta}
    The unknown parameters $\theta$ belong to a known compact set $\Theta$.
    Furthermore, the parameter estimation error $\tilde{\theta} \triangleq \hat{\theta} - \theta$ also belongs to a known compact set $\tilde{\Theta}$.
\end{assumption}
\begin{assumption}
    \label{assumption:error}
    The model parameter estimate $\hat{\theta}$ is being generated online via a suitable model estimator where the estimation error $\tilde{\theta}$ is monotonically non-increasing.
\end{assumption}
Assumptions~\ref{assumption:theta} and~\ref{assumption:error} are standard in the indirect adaptive control and adaptive MPC literature \cite{khalil1996adaptive, kohler2021robust}.
Note that Assumption~\ref{assumption:error} is more restrictive but can be generally achieved using least squares over a sliding window \cite{chowdhary2013concurrent} or set membership identification \cite{kosut1992set}
.

The goal of this work is to develop constructive conditions for a model-parameterized feedback policy $\kappa_{\theta}:\mathbb{R}^n \times \mathbb{R}^n \times \mathbb{R}^p \rightarrow \mathbb{R}^m$ such that the state $x(t)$ converges exponentially to a region near a desired trajectory $x_d(t)$ when the system dynamics are uncertain and being estimated online.
Note that $\kappa_{\theta}$ must be model-parameterized so that the tracking error converges to zero as the parameter estimate converges to its true value.
Since this work is concerned with converging to a time-varying trajectory despite the presence of model error, it is beneficial to introduce the notion of an input-to-state convergent system.
\begin{definition}[\cite{pavlov2006uniform}]
\label{def:main}
    A closed-loop system is \emph{input-to-state convergent} (ISC) with respect to the input $w(t)$ if, for any feasible trajectory $x_d(t)$, the state $x(t)$ satisfies
    \begin{equation*}
        |x(t)-x_d(t)| \leq \beta(|x(t_0)-x_d(t_0)|,t) + \sigma\Biggl( \sup_{\tau \in [t_0,t]} |w(\tau)|\Biggr),
    \end{equation*}
     where $\beta: \mathbb{R} \times \mathbb{R} \rightarrow \mathbb{R}$ is a class $\mathcal{KL}$ function and $\sigma: \mathbb{R} \rightarrow \mathbb{R}$ is a class $\mathcal{K}$ function.    
\end{definition}
\cref{def:main} ensures that the tracking error transients is bounded, and that $x(t) \rightarrow x_d(t)$ if $w(t) \rightarrow 0$.
More broadly, ISC \cite{pavlov2006uniform} is a natural extension to ISS \cite{sontag1995characterizations} and is deeply related to $\delta$ISS \cite{angeli2002lyapunov}.
In particular, ISC and $\delta$ISS are equivalent if the system evolves on a compact set \cite{ruffer2013convergent}.
We will also make use of a specialized version of \cref{def:main}, namely the notion of an exponentially convergent system, where the $\mathcal{KL}$ function $\beta$ takes an exponential form.
\begin{definition}
\label{def:main_exp}
    A closed-loop system is \emph{exponentially input-to-state convergent} (eISC) with respect to the input $w(t)$ if, for any feasible trajectory $x_d(t)$, state $x(t)$ satisfies
    \begin{equation*}
        \begin{aligned}
            |x(t)-x_d(t)|  \leq \  & \sigma_1(|x(t_0)-x_d(t_0)|) e^{-\lambda {(t-t_0)}} \\ 
             &+ \sigma_2\Biggl( \sup_{\tau \in [t_0,t]} |w(\tau)|\Biggr),
        \end{aligned}
    \end{equation*}
     where $\lambda \in \mathbb{R}_{>0}$ is the convergence rate and $\sigma_{1,2}: \mathbb{R} \rightarrow \mathbb{R}$ are class $\mathcal{K}$ functions.    
\end{definition}

It is important to note that eISC and exponential $\delta$ISS are not necessarily equivalent as exponential convergence to a single trajectory is a weaker condition than exponential convergence of two arbitrary trajectories. 
With that said, it is easy to show that exponential $\delta$ISS implies eISC -- a fact we will use later in this work.

\begin{remark}
Any system that satisfies \cref{def:main} or \cref{def:main_exp} also satisfies the Universal $L_\infty$ gain definition as presented in \cite{zhao2022tube}.
The key difference between the definition presented here and that in \cite{zhao2022tube} is that \cref{def:main,def:main_exp} use a generic class $\mathcal{K}$ function $\sigma$ and ensure that $\beta$ is a $\mathcal{KL}$ function by definition or by making it a decaying exponential.
\end{remark}

\cref{def:main,def:main_exp} are important characterizations of closed-loop systems but do not contain constructive conditions for a control law or control Lyapunov function (clf) that satisfy their respective mathematical definitions.
While $w(t)$ is a general input, within the context of indirect adaptive control one aims to achieve ISC with respect to model estimation error while the model estimate is being updated online.
A necessary and sufficient condition for a system to be eISC with respect to $\tilde{\theta}$ and $\tfrac{d}{dt}{\hat{\theta}}$ is the existence of an exponentially ISC-control Lyapunov function (eISC-clf)\footnote{An eISC-clf is similar to an ISS-clf discussed in \cite{krstic1995nonlinear} but instead ensures exponential convergence to a time-varying trajectory $x_d(t)$. Proving the equivalence between eISC and the existence of an eISC-clf follows similarly to that in \cite{krstic1995nonlinear} and is omitted for brevity}.

\begin{definition}
    \label{def:isc-clf}
    A continuously differentiable function $V_{\theta} : \mathbb{R}^n \times \mathbb{R}^n \times \mathbb{R}^p \rightarrow \mathbb{R}_+$ where $V_{\theta}(x,x_d) = 0 \iff x(t) = x_d(t)$ is an \emph{exponentially input-to-state convergent control Lyapunov function} (eISC-clf) if the following hold
    \begin{gather}
        k_1 |x-x_d|^a \leq V_{\theta}(x,x_d) \leq k_2 |x-x_d|^a \label{eq:isc-clf-pd} \\
        |x(t) - x_d(t)| \geq \sigma\left( \sup_{\tau \in [t_0,t]} \left| \left[ \begin{array}{c} {\tilde{\theta}(\tau)} \\ \dot{\hat{\theta}}(\tau) \end{array} \right] \right| \right) \implies  \nonumber \\
        \underset{u \in \mathbb{R}^m}{\mathrm{inf}} \Big\{ \nabla_x V_{\hat{\theta}}^\top F_{\hat{\theta}}(x,u) + \nabla_{x_d} V_{\hat{\theta}}^\top F_{\hat{\theta}}(x_d,u_d) \nonumber \\
        - \nabla_x V_{\hat{\theta}}^\top \Delta^\top \tilde{\theta} + \nabla_{\hat{\theta}} V_{\hat{\theta}}^\top \dot{\hat{\theta}} \Big\} \leq -k_3 |x-x_d|^a \label{eq:isc-clf}
    \end{gather}
    where $k_1,\,k_2,\,k_3,\,a\in \mathbb{R}_{>0}$, $\sigma:\mathbb{R}\rightarrow\mathbb{R}$ is a class $\mathcal{K}$ function, and $\hat{\theta} \in \Theta$ and $\tfrac{d}{dt}{\hat{\theta}} \in \Omega$ belong to compact sets.
\end{definition}

Synthesizing an eISC-clf that is robust to model estimation error and estimation transients is not trivial due to the presence of $\nabla_{\hat{\theta}} V(x)^\top \, \tfrac{d}{dt}{\hat{\theta}}$ in \cref{eq:isc-clf}.
Techniques have been developed for systems in strict-feedback form \cite{krstic1995nonlinear,zamani2011backstepping} but a systematic constructive procedure for general nonlinear systems has yet to be proposed.
A promising approach, and the one taken here, is the use of contraction analysis \cite{lohmiller1998contraction} and the control contraction metric (ccm) \cite{manchester2017control}.
In \cite{manchester2018robust}, the robust ccm was introduced to ensure eISC despite the presence of model uncertainties or external disturbances; the robust ccm was later used for robust motion planning in \cite{zhao2022tube}.
The main limitation of using the robust ccm framework is that the tracking error is not guaranteed to converge to zero even if the model is perfectly known after online estimation.
This is a direct result of the metric (and hence controller) being designed for a \emph{single} model realization.
In essence, picking a single model for metric / controller design eliminates the ability to incorporate new model information gathered online without violating tracking error or transient bounds.
The method proposed in this work addresses this limitation while guaranteeing that the tracking error transients is bounded \emph{and} that the tracking error will converge to zero if the model estimation error converges to zero.

\section{Main Result}
\label{sec:proofs}
\subsection{Overview}
This section contains the main results of this paper. 
First, a brief review of control contraction metrics is presented.
Then, the adaptive robust ccm is defined.
Lastly, we show how an adaptive robust ccm can be constructed via convex optimization to achieve the bounds in \cref{def:main_exp}.

\subsection{Adaptive Robust Control Contraction Metrics}

Contraction analysis \cite{lohmiller1998contraction} is an attractive framework for synthesizing stabilizing control laws.
In particular, the search for a ccm $M:\mathbb{R}^n\rightarrow 
\mathcal{S}_+^n$ can be formulated as a convex condition amenable to numerical methods.
Furthermore, once a metric $M(x)$ is found then the corresponding control law ensures any two trajectories converge to each other exponentially. 
This convergence can be stated in terms of geodesics and Riemannian energy. 
For a smooth manifold $\mathcal{M}$ a geodesic $\gamma: [0,1] \times \mathbb{R} \rightarrow \mathbb{R}^n$ is the shortest path between any two points.
If we impose the boundary conditions $\gamma(0,t)=x_d(t)$, and $\gamma(1,t)=x(t)$ then the Riemannian energy $E(\gamma(s,t))=\int_0^1 \gamma_s(s,t)^\top M(\gamma(s,t))\gamma_s(s,t)\,ds$  where $\gamma_s(s,t)\triangleq \nabla_s \gamma(s,t)$ can be interpreted as a form of tracking error since $E(\gamma(s,t)) = 0 \iff x(t) = x_d(t)$.
Since $M(x)$ is a ccm then $\dot{E}(\gamma(s,t)) \leq - 2 \lambda E(\gamma(s,t))$ yielding $x(t) \rightarrow x_d(t)$ exponentially with rate $\lambda$.

To further motivate this work, and the need for a new type of ccm, consider the differential dynamics of \cref{eq:system} $\dot \delta_x= A_{\theta}(x,u) \delta_x + \Delta(x)^\top \delta_\theta + B(x) \delta_u,$
where $ A_{\theta}(x,u)= \nabla_x f(x) + \sum_{i=1}^{m}{\nabla_x b_i(x) \, u_i}+\sum_{i=1}^p{\nabla_x\Delta_i(x)\theta_i}$
is the Jacobian of the dynamics.
When the model is known, a metric $M(x)$ is constructed so that the implication $\delta_x^\top M B = 0 \implies \delta_x^\top (A^\top M + M A + \dot{M})\delta_x \leq -2 \lambda \delta_x^\top M \delta_x$ is true; the dependency on $u$ in the Jacobian can be eliminated by imposing $B(x)$ be a Killing vector for $M(x)$.
However, the presence of the uncertainty in the differential dynamics and Jacobian requires the metric satisfy a stronger condition to achieve bounded transients.
The robust ccm was proposed to meet this need but is unable to incorporate online model estimates thereby limiting its performance.
Hence, a new metric is needed in order to incorporate online model estimates without sacrificing tracking error transient bounds.
\begin{definition}
    A uniformly bounded Riemannian metric $M_\theta: \mathbb{R}^n \times \mathbb{R}^p \rightarrow \mathcal{S}^n_+$ is an \emph{adaptive robust control contraction metric} (arccm) if for each $\theta \in \Theta$ the dual metric $W_\theta(x)\triangleq M_\theta(x)^{-1}$ satisfies 
    \begin{align}
    & -\partial_{{x}} W_\theta + A_\theta W_\theta + W_\theta A_\theta^\top + BY_\theta + Y_\theta^\top B^\top  \nonumber \\
    & \hphantom{-} \preceq - 2 \lambda W_\theta + \frac{1}{\alpha^2}\Delta^\top \Delta,  \tag{C1} \label{eq:arccm_c1} \\[5pt]
    & \partial^\circ_{b_i} W_\theta - W_\theta \nabla_x b_i^\top -  \nabla_x b_i W_\theta = 0,~~ i=1,\dots,m \tag{C2} \label{eq:arccm_c2} \\[5pt]
    &-\mu W_\theta \preceq \nabla_{\theta_i} W_\theta \preceq \mu W_\theta ,~~ i=1,...,p  \tag{C3} \label{eq:arccm_c3}
    \end{align}
    where $Y_\theta:\mathbb{R}^n \times \mathbb{R}^p \rightarrow \mathbb{R}^m \times \mathbb{R}^n $ and $\lambda,\,\mu,\,\alpha \in \mathbb{R}_{>0}$.
    \label{def:arccm}
\end{definition}

Several remarks are in order. 
Firstly, \cref{eq:arccm_c1,eq:arccm_c2} are similar to those in \cite{manchester2018robust} but the metric is now allowed to be parameter-dependent. 
In the absence of parameter adaptation \cref{eq:arccm_c1,eq:arccm_c2} are sufficient conditions for the closed-loop system to be ISC with respect to $\tilde{\theta}$.
Note \cref{eq:arccm_c2} is only needed if the control input matrix depends on $x$ thereby ensuring $B(x)$ is a Killing vector field for $M(x)$.
Secondly, since the metric is parameter-dependent, a valid metric is immediately available for any $\theta \in \Theta$.
Hence, $M_{\theta}(x)$ represents a \emph{family of metrics} that can immediately incorporate online estimates of $\theta$, thereby addressing the fundamental limitations of the robust ccm.
Since the metric will change as $\theta$ is continuously being replaced with $\hat{\theta}$ from the estimator, \cref{eq:arccm_c1} would usually include another term related to the transients of $\hat{\theta}$ thereby complicating the search for $M_{\theta}(x)$.
Introducing \cref{eq:arccm_c3}, which is part of the novelty of this work, we bypass the difficulty of incorporating model estimation transients in \cref{eq:arccm_c1} while still achieving tracking error transient bounds; this will be shown in \cref{thm:main}. 
Condition \cref{eq:arccm_c3} can be equivalently stated (see \cref{prop:logE}) as a bound on the log gradient of the Riemannian energy when $x(t)\neq x_d(t)$.
This interpretation of condition \cref{eq:arccm_c3} provides a useful insight into how \cref{eq:arccm_c3} influences the metric and subsequently the Riemannian energy, namely it limits how fast the log of the Riemannian energy can change.

\begin{proposition}
\label{prop:logE}
When $x(t)\neq x_d(t)$ then condition \cref{eq:arccm_c3} holds for any curve $c(s,t)$  if and only if 
\begin{equation}
    |\nabla_{\theta_i} \log (E_\theta(c(s,t))| \leq \mu ,~~ i= 1,..., p \label{eq:arccm_c3b}
\end{equation}
where $E_\theta(c(s,t))$ is the Riemannian energy of curve $c(s,t)$.
\end{proposition}
\begin{proof}
Since $x(t) \neq x_d(t)$ then $E_\theta(c) > 0$.
Suppose $|\nabla_{\theta_i} \log (E_\theta(c))| \leq \mu$ (omitting the arguments of curve $c$).
Applying the chain rule and Liebnitz integral rule to the definition of the Riemannian energy, we obtain
\begin{equation*} 
\begin{aligned}
    & | \nabla_{\theta_i} \log (E_\theta(c)) | = |\nabla_{\theta_i} E_\theta(c)| \, / \, E_\theta(c) \leq \mu, \\[5pt]
    & \implies -\mu \, E_\theta(c) \leq \nabla_{\theta_i} E_\theta(c) \leq \mu \, E_\theta(c), \\[5pt]
    & \implies -\mu \, M_\theta(c) \preceq \nabla_{\theta_i} M(c) \preceq \mu \, M_\theta(c),
\end{aligned}
\end{equation*}
where the last implication follows from the definition of $E_\theta$.
Letting $W_\theta(x)=M_\theta(x)^{-1}$ then one arrives to \cref{eq:arccm_c3}.
Conversely, if \cref{eq:arccm_c3} holds then $|c_s^\top \nabla_{\theta_i} M(c) c_s| \leq \mu \,  c_s^\top M_\theta(c) c_s$ for any curve $c$.
Integrating with respect to $s$ and applying Cauchy-Schwarz inequality yields
\begin{equation*}
\begin{aligned}
    |\nabla_{\theta_i}E_\theta(c)| \leq \mu \, E_\theta(c) \implies |\nabla_{\theta_i} \log (E_\theta(c))| \leq \mu. \qedhere
\end{aligned}
\end{equation*}
\end{proof}

For numerically constructing an arccm, we see that \cref{eq:arccm_c1,eq:arccm_c2,eq:arccm_c3} are jointly convex in $W_\theta$, $Y_\theta$ and $\frac{1}{\alpha^2}$ so they can be efficiently computed using numerical methods.
However, they are non-convex in $\mu$ or $\lambda$ so a line search method is needed to optimize over these parameters. 
Picking an appropriate cost function that balances the competing objectives of convergence rate, overshoot, and control effort is also critical to obtain the desirable closed-loop response.

We now present the main result of the paper, which shows that the tractable conditions laid out in \cref{def:arccm} are sufficient to achieve definition \cref{def:main_exp}. 
The methodology of the proof is similar to the ones in \cite{manchester2018robust, zhao2022tube} with the key difference being how to handle $M_{\hat{\theta}}$ being a function of $\hat \theta$.
If $x(t)$ happens to be in the cut locus then the proof can be modified with the technique used in \cite{singh2017robust}.  
We make the following assumption to ensure the uniqueness of geodesics and smoothness of the Riemannian energy.

\begin{assumption}
For the control system \cref{eq:system} the set of times for which $x(t)$ is in the cut locus of $x_d(t)$ has zero measure. As shown in \cite{singh2017robust} this assumption is not strictly necessary and the upper Dini derivative can be used to remove the assumption. 
\end{assumption}

\begin{theorem}
\label{thm:main}
    Let $\hat{\theta}$ be the current estimate of the unknown model parameters $\theta$.
    If an adaptive robust ccm $M_{\theta}(x)$ exists, then the closed-loop system is eISC with respect to the estimation error $\tilde{\theta}$ with the control law
    \begin{equation}
    \label{eq:control}
        u(t) = u_d(t) + \int_0^1 K_{\hat{\theta}}(\gamma(s,t)) \, \gamma_s(s,t) \, ds,
    \end{equation}
    where $K_{\hat{\theta}}(x) \triangleq  Y_{\hat{\theta}}(x) M_{\hat{\theta}}(x)$ computed via the conditions in \cref{def:arccm} and $\gamma(s,t)$ is a geodesic connecting $x(t)$ and $x_d(t)$. Hence, if $\tilde{\theta}\rightarrow$ then $x(t)\rightarrow x_d(t)$ exponentially.
\end{theorem}
\begin{proof}
    Let $\gamma(s,t)$ be a geodesic where $\gamma(0,t) = x_d(t)$ and $\gamma(1,t) = x(t)$.
    Moreover, let $c(s,t)$ be a curve that evolves according to \cref{eq:system} where $c(s,t') = \gamma(s,t')$ for some time $t'$.
    If the control law $u(s,t)=u_d(t)+\int_0^s{K_{\hat{\theta}}(c)c_s\, ds}$ where $\delta_u = u_s(s,t) = K_{\hat{\theta}}(c)c_s$ is applied for some interval $t\in [t', t'+\epsilon]$, then the differential dynamic of $c$ becomes
    \begin{equation*}
        \dot c_s = \left[{A}_{\hat{\theta}}(c,u) + B(c) K_{\hat{\theta}}(c) \right]c_s + \Delta(c)^\top \tilde{\theta},
    \end{equation*}
    where $\theta(s,t)=s\hat\theta(t)+(1-s)\theta(t)$ so $\delta_{\theta}=\theta_s=\tilde{\theta}(t)$.
    Letting $V_{\hat{\theta}}(c,c_s) = c_s^\top M_{\hat{\theta}}(c) c_s$ then
    \begin{equation*}
        \begin{aligned}
        \frac{d}{dt}(c_s^\top M_{\hat{\theta}} c_s) & = c_s^\top \Big( A^\top M_{\hat{\theta}} + M_{\hat{\theta}} A + \dot{M}_{\hat{\theta}} + \tilde{\theta}^\top \Delta M_{\hat{\theta}}  \\
        & \hphantom{=} + M_{\hat{\theta}} \Delta^\top \tilde{\theta}  + K_{\hat{\theta}}^\top g^\top M_{\hat{\theta}} + M_{\hat{\theta}} K_{\hat{\theta}} B \Big) c_s,
        \end{aligned}
    \end{equation*}
    where $\dot{M}_{\hat{\theta}}(c) = {\partial}_{{c}} M_{\hat{\theta}}(c) + {\partial}_{{\hat{\theta}}} M_{\hat{\theta}}(c)$.
    Making use of \cref{eq:arccm_c1} in \cref{def:arccm}, taking the Schur compliment and multiplying both sides by $[\eta_x, \delta_\theta]$ and $[\eta_x, \delta_\theta]^\top$ gives
 \begin{equation*}
    \begin{bmatrix} \eta_x\\ \delta_\theta \end{bmatrix}^\top \begin{bmatrix} \mathcal{W} & -\Delta \\ 
    -\Delta & \alpha^2 I
    \end{bmatrix} \begin{bmatrix} \eta_x  \\ \delta_\theta \end{bmatrix}  \geq 0,
\end{equation*}
    where $\mathcal{W}=\partial_c W_\theta - \lambda W_\theta - A_\theta W_\theta - W_\theta A_\theta^\top  - g Y_\theta - Y_\theta^\top g^\top$.
    Expanding the above expression, making the substitutions $\eta_x = M_{\hat{\theta}} c_s$ and $Y_{\hat{\theta}} =K_{\hat{\theta}} M_{\hat{\theta}}^{-1}$, and applying condition \cref{eq:arccm_c3} one can show that $\dot{V}_{\hat{\theta}}$ becomes
    \begin{equation*}
        \frac{d}{dt}(c_s^\top M_{\hat{\theta}} c_s)\leq -\rho(t) c_s^\top M_{\hat{\theta}} c_s + \alpha^2 |\tilde{\theta}|^2,
    \end{equation*}
    where $\rho(t) \triangleq \lambda - p\mu |\tfrac{d}{dt}{\hat \theta}|$. 
    Integrating along $c(s,t)$ leads to
    \begin{equation*}
        \int_{0}^{1}{\frac{d}{dt}(c_s^\top M_{\hat{\theta}} c_s)ds}\leq  -\rho(t)\int_{0}^{1}c_s^\top M_{\hat{\theta}} c_s ds+\alpha^2\int_{0}^{1}|\tilde{\theta}|^2 ds.
    \end{equation*}
    Switching the order of integration and differentiation yeilds
    \begin{equation*}
        \frac{d}{dt}E(c(s,t))\leq -\rho(t)E(c(s,t))+\alpha^2|\tilde \theta|^2.
    \end{equation*}
    For $t \in [t',t'+\epsilon]$ taking the limit as $\epsilon$ goes to $0$ gives
    \begin{align*}
        \lim_{\epsilon\rightarrow0} \frac{d}{dt}E(c(s,t)) &= \frac{d}{dt}\lim_{\epsilon\rightarrow0} E(c(s,t)) \\ 
        &\leq \lim_{\epsilon\rightarrow0} \left( -\rho(t)E(c(s,t))+\alpha^2|\tilde \theta|^2 \right).
    \end{align*}
    Since $c(s,t')=\gamma(s,t')$ and $t'$ was arbitrary we obtain
     \begin{equation}
        \frac{d}{dt} E(\gamma(s,t))\leq  -\rho(t)E(\gamma(s,t))+\alpha^2|\tilde \theta|^2.
        \label{eq:energy_eq_dot}
    \end{equation}
    Recalling $\tilde \theta$ and $\tfrac{d}{dt}{\hat \theta}$ are bounded by assumption, letting $\bar{\rho} \triangleq \lambda-p\mu ||\tfrac{d}{dt}{\hat \theta}||_\infty \leq \rho$ then \cref{eq:energy_eq_dot} becomes
    \begin{multline}
        E(\gamma(s,t))\leq E(\gamma(s,0))e^{-\bar{\rho} t} + \frac{\alpha^2}{\bar \rho }||\tilde\theta||^2_\infty(1-e^{- \bar \rho t}),
        \label{eq:energy_bound}
    \end{multline}
    which can be put into the same form as that in \cref{def:main_exp} by noting $M_{\hat{\theta}}(x)$ is uniformly bounded so there exists $\ubar{a},\,\bar{a} \in \mathbb{R}_{>0}$ such that $\ubar{a} I \preceq M_{\hat{\theta}}(x) \preceq \bar{a} I$ which leads to
    \begin{equation*}
        |x(t)-x_d(t)| \leq \sqrt{\tfrac{\bar{a}}{\ubar{a}}}|x(0)-x_d(0)|e^{-\bar{\rho} t} + \frac{\alpha^2}{\bar \rho }||\tilde\theta||^2_\infty(1-e^{- \bar \rho t}).
    \end{equation*}
    Therefore, the closed-loop system is eISC with respect to $\tilde{\theta}$ with the adaptive robust ccm $M_{\theta}(x)$ and control law \cref{eq:control}.
    Moreover, if $\tilde{\theta} \rightarrow 0$ then $x(t)\rightarrow x_d(t)$ exponentially. \qedhere
\end{proof}

It is important to note that although the system no longer has to be ISC / eISC with respect to the estimation transients, the convergence rate of the closed-loop system will change during estimation transients. 
It is critical to ensure that $\tfrac{d}{dt}\hat{\theta}$ never violates the condition $\rho(t) = \lambda - p\mu |\tfrac{d}{dt}{\hat{\theta}}| > 0$.
Pragmatically this can be achieved by selecting or modulating the parameter estimation rate.
Furthermore, there is an inherent trade-off between short-term tracking error convergence and long-term control performance from having a more accurate model estimate. 
In other words, the effects of fast model estimation must be carefully considered because large estimation transients will result in slower short-term convergence but better long-term control since the model is more accurately known.

\subsection{Offline / Online Computation}
\label{sub:imp}
As mentioned previously, the conditions in \cref{def:arccm} are jointly convex in $W_{\theta}(x),\, Y_{\theta}(x),\,\mathrm{and}\,\tfrac{1}{\alpha^2}$.
The values of these are determined offline using existing state-dependent LMI  solvers. 
The controller \cref{eq:control} needs to be computed online, and requires a geodesic at each time.
Geodesic computation does involve a nonlinear optimization but can be performed in real-time \cite{leung2017geodesic} and is less expensive than nonlinear MPC.
The parameter estimates $\hat{\theta}$ should be computed online by any suitable model estimator to maximize performance.

\label{sec:results}
\section{Illustrative Example}
Consider the system with state $x=[x_1,~ x_2,~ x_3]^\top$ unknown parameters $\theta=[\theta_1, ~ \theta_2, ~ \theta_3, ~ \theta_4]^\top$ and dynamics
\begin{equation}
    \begin{bmatrix}
        \dot x_1 \\
        \dot x_2 \\
        \dot x_3
    \end{bmatrix} = \begin{bmatrix}x_3-\theta_1 x_1 \\ -x_2-\theta_2 x_1^2 \\ \tanh(x_2)- \theta_3 x_3 - \theta_4 x_1^2\end{bmatrix}  + \begin{bmatrix}
        0 \\ 0 \\ 1
    \end{bmatrix}u,
    \label{example_sys}
\end{equation}
where $\theta_1 \in [-1, ~ 1]$, $\theta_2 \in [0.5, ~ 1.5]$,  $\theta_3 \in [-0.6, ~ 0.75]$ and $\theta_4 \in [-1.75, ~ 0.5]$. 
The true model parameters are $\theta^*=[-0.3, ~ 0.8, ~ -0.25, ~ -0.75]^\top$.
This system has been used before in \cite{lopez2021universal} (modified from that in \cite{manchester2017control}) because it demonstrates the versatility of contraction since the systems is not feedback linearizable, is not in strict feedback form, and has unmatched uncertainties.
A desired trajectory $x_d$ was generated by specifying $x_{1d}(t)=\sin(t)$; the remaining states $x_{2d}(t)$ and $x_{3d}(t)$ were generated based on an updating model of the system.
An arccm was computed via sum-of-squares programming with YALMIP \cite{lofberg2004yalmip} where $W_\theta$ and $Y_\theta$ were parameterized by fourth-order matrix polynomials in state $x$ and parameter $\theta$.
In order to ensure the conditions in \cref{def:arccm} were satisfied for all $\theta \in \Theta$, \cref{eq:arccm_c1,eq:arccm_c3} were imposed at each point in a uniform 21 x 21 grid for the unmatched parameters $\theta_1$ and $\theta_2$. 
The control input was computed as outline in \cref{sub:imp}.

\begin{figure}[t!]
    \centering
    \includegraphics[width=.45\textwidth]{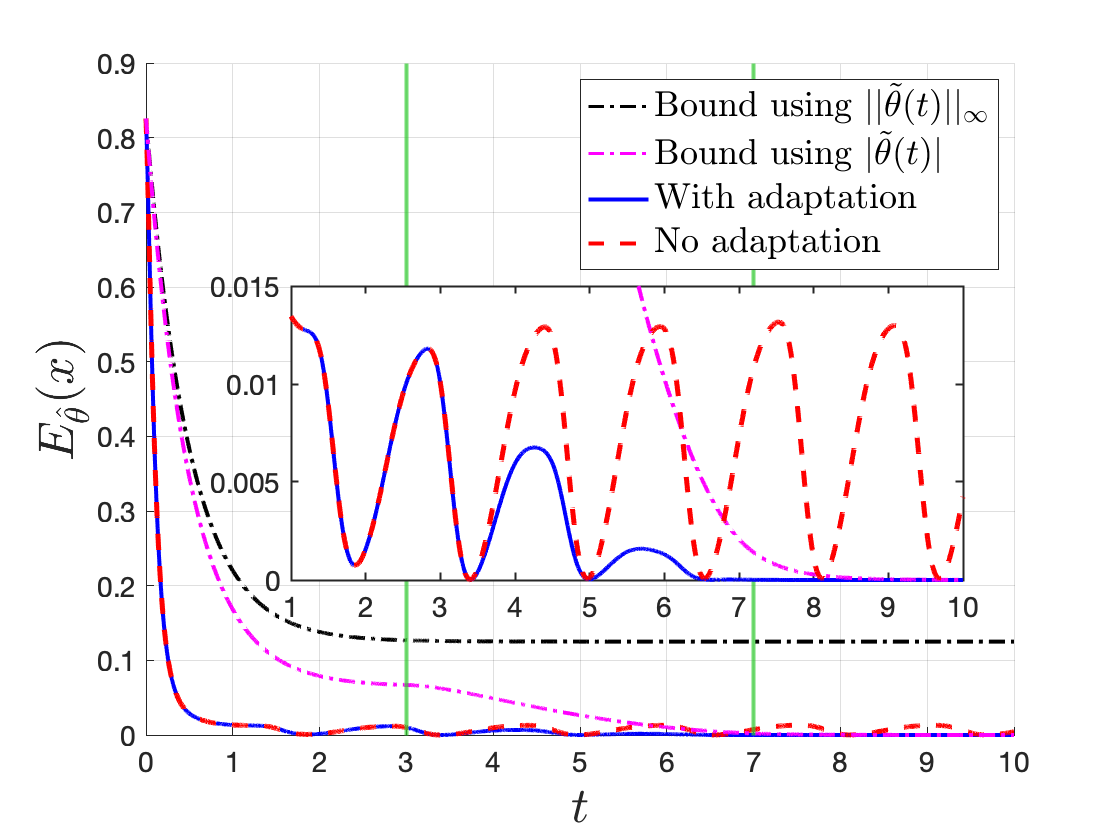}
    \caption{ Riemannian energy of the new method is lower than that of a controller with a static reference model. The model with adaptation begins adapting at $t=3$ s and finishes at $t=7$ s.}
    \label{fig:energy}
    \vskip -0.2in
\end{figure}

\begin{figure}[t!]
    \begin{subfigure}{.49\columnwidth}
         \centering
         \includegraphics[trim={30 10 100 0},clip, width=1\textwidth]{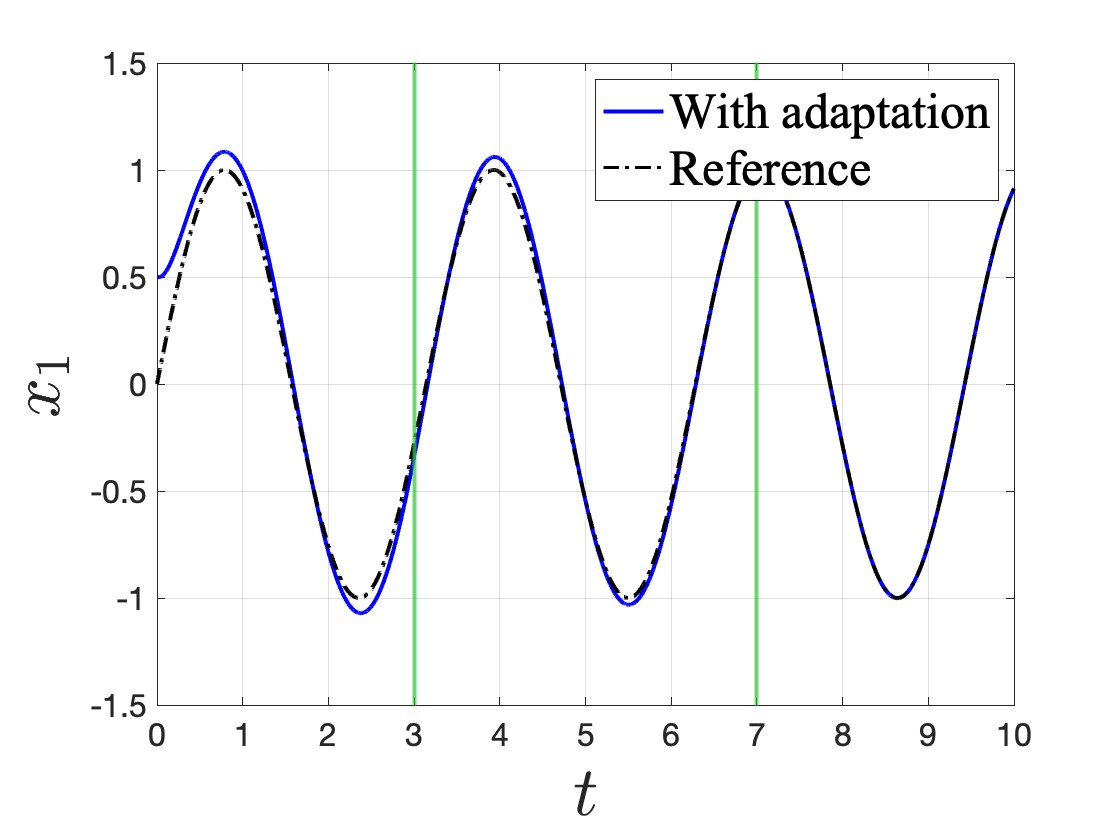}
         \caption{$x_1$ with adaptation}
         \label{fig:ya_x1}
     \end{subfigure}
     \hfill{}
     \begin{subfigure}{.49\columnwidth}
         \centering
         \includegraphics[trim={30 10 100 0},clip,width=1\textwidth]{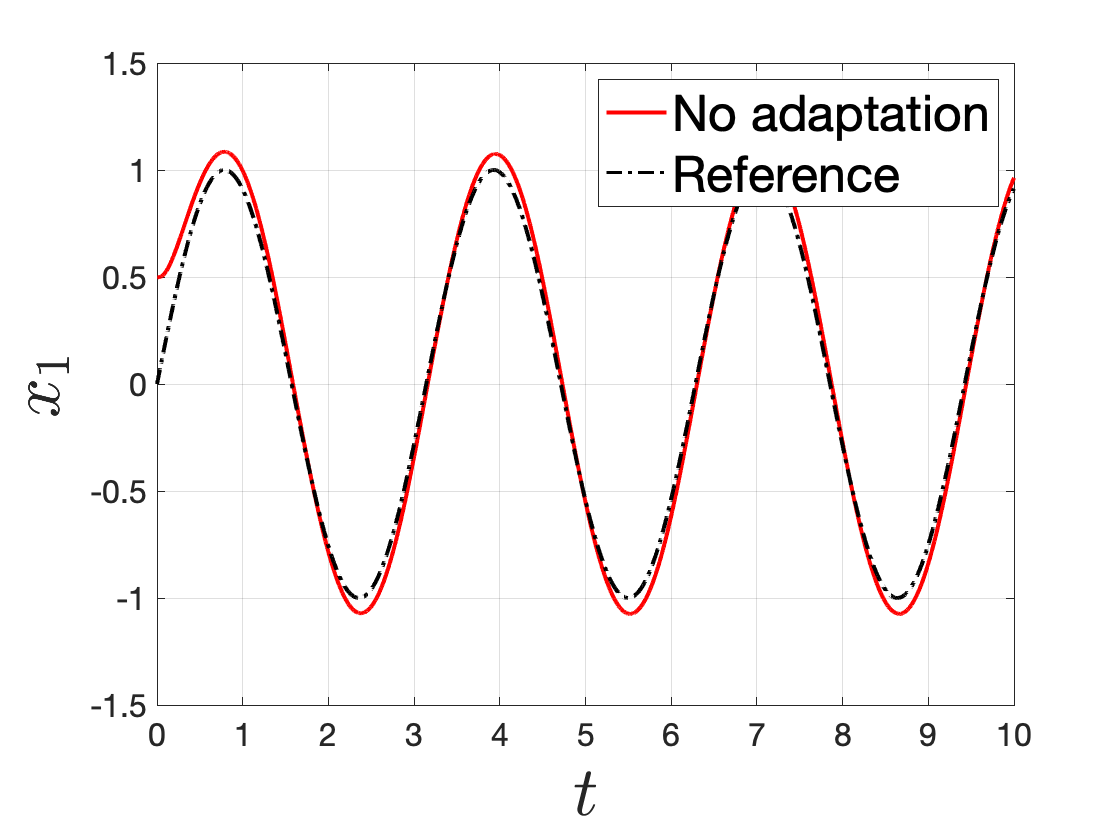}
         \caption{$x_1$ without adaptation}
         \label{fig:na_x1}
     \end{subfigure}
     \vfill{}
     \begin{subfigure}{.49\columnwidth}
         \centering
         \includegraphics[trim={30 10 100 0},clip,width=1\textwidth]{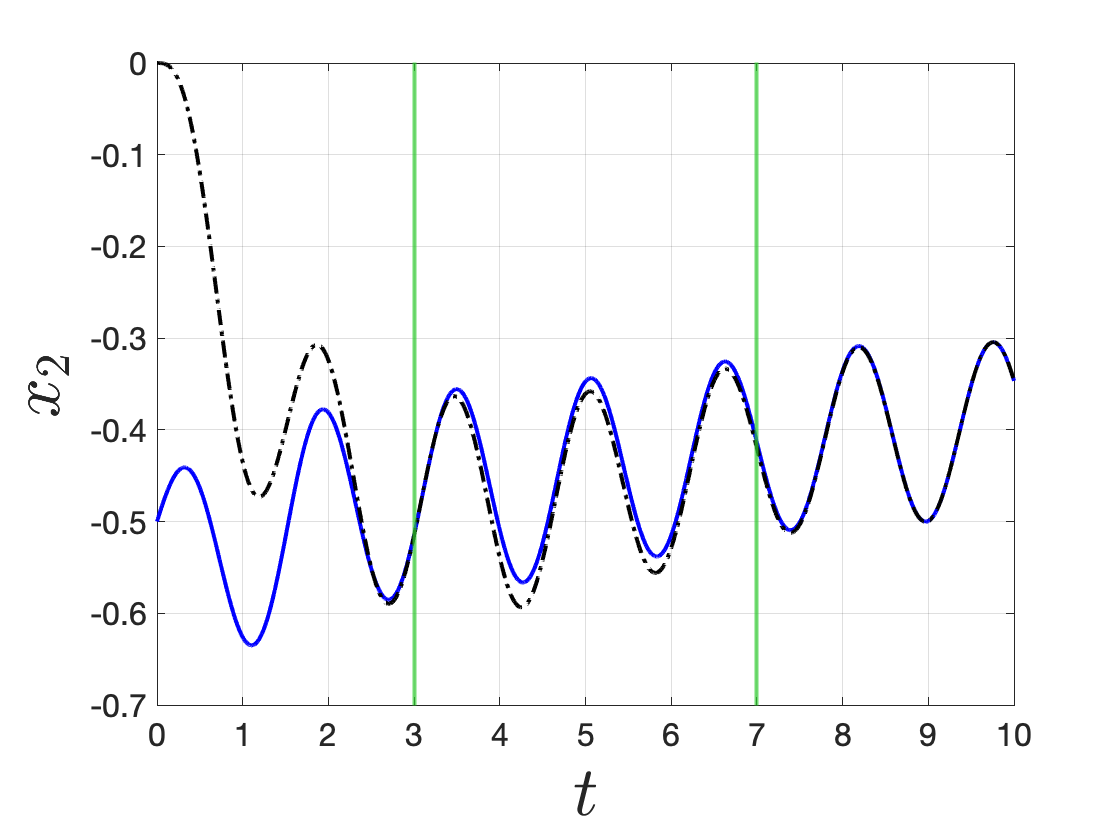}
         \caption{$x_2$ with adaptation}
         \label{fig:ya_x2}
     \end{subfigure}
     \hfill{}
     \begin{subfigure}{.49\columnwidth}
         \centering
         \includegraphics[trim={30 10 100 0},clip,width=1\textwidth]{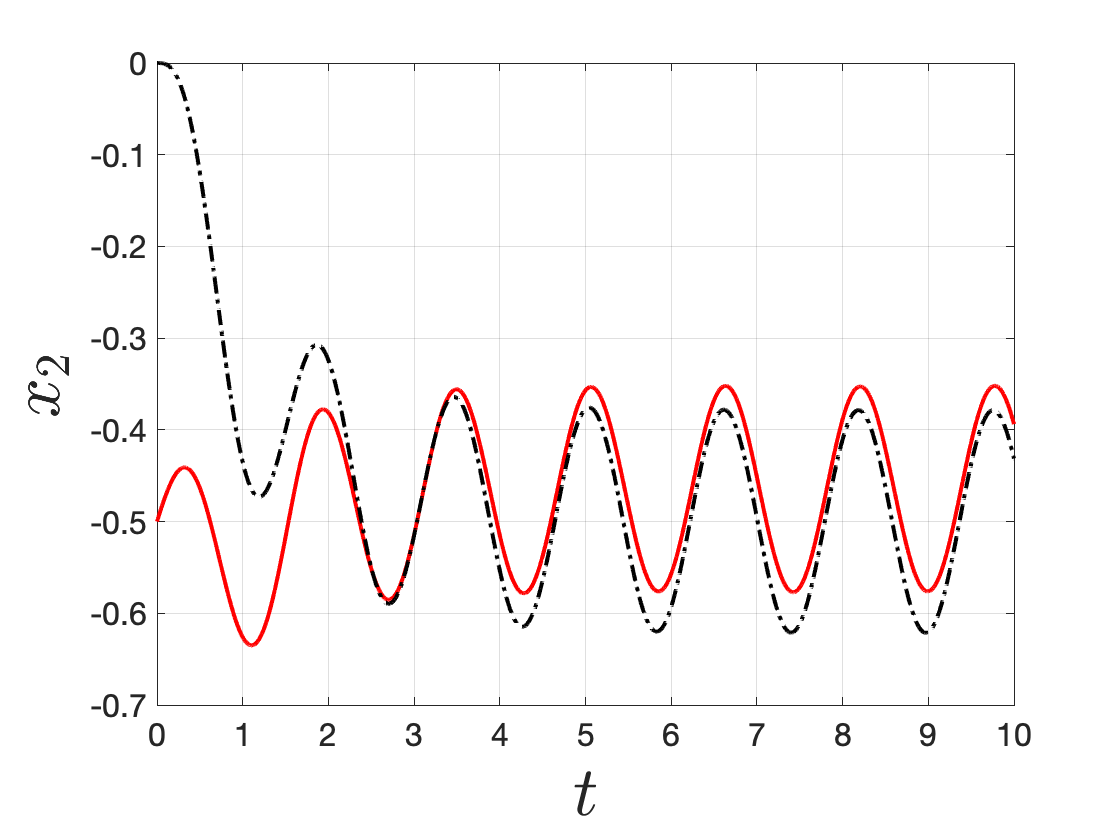}
         \caption{$x_2$ without adaptation}
         \label{fig:na_x2}
     \end{subfigure}
     \vfill{}
     \caption{States with and without model updates. The reference for $x_1$ is the same for both controllers but differs for other states due to model updates. The model updates start at $t=3$ s and converges at $t=7$ s (vertical green lines). State $x_3$ omitted for brevity as it exhibits identical behavior.}
     \vskip -0.3in
     \label{fig:states}
\end{figure}

Two controllers were tested to verify the properties of the proposed approach: one with model estimation and the other without. 
The adaptive controller was initialized with an estimate $\hat \theta$ set to the midpoint of $\Theta$.
At $t=3$ s the parameters began to linearly change to the true parameters with convergence occurring at $t=7$ s, i.e., $\tilde{\theta}=0$; this is indicated by green vertical lines in \cref{fig:energy,,fig:states}. 
The non-adaptive controller used the same metric but $\hat \theta$ was fixed to the midpoint of $\Theta$. 
 The Riemannian energy --- an indicator of the closed-loop tracking error --- was calculated for the two controllers and is plotted in \cref{fig:energy}.
 As seen in the figure both controllers obey \cref{eq:energy_bound}, as the energy is always upper bounded by the predicted conservative bound (black line). 
 The black line is the bound based on the most aggressive adaption and worst case $\tilde \theta$, which can be used to construct a worst-case error bound from the start and therefore be employed for predictive control. 
 \cref{fig:energy} also shows the bound calculated by integrating  \cref{eq:energy_eq_dot} (magenta line).
 In practice, the bound from \cref{eq:energy_eq_dot} would not be known as it relies on knowing $\tilde \theta(t)$ (and hence knowing $\theta$ at all time) but is useful here to confirm the theoretical result that $x \rightarrow x_d$ as $\tilde{\theta} \rightarrow 0$.
 As shown in the inset of \cref{fig:energy}, the tracking error with the robust (non-adaptive) controller satisfies the conservative bound but does not converge to zero. 
 Although the magenta bound is not an achievable worst case bound, the worst case bound could be improved if $\tilde \Theta$ is reduced via, e.g., set membership identification, to achieve zero tracking error. 
 This is in contrast to only improving $\tilde \theta$ which yields better performance but the conservative bound is unchanged.
 \cref{fig:states} shows the time history of $x_1$ and $x_2$ for reference.

\section{Conclusion}
\label{sec:conclusion}
We presented a contraction-based framework for robust trajectory tracking despite the presence of uncertain dynamics that are being estimated online.
After introducing the adaptive robust control contraction metric (arccm), which is allowed to depend on the unknown parameters, a feedback control law was proposed and shown to yield an exponentially input-to-state convergent closed-loop system. 
By allowing the metric to depend on the unknown parameters we are able to guarantee the trajectory tracking error will converge to zero if the parameter estimates converge to their true value.
The key mechanism for allowing the metric to be parameter dependent is a novel gradient condition that only changes the convergence rate, as opposed to sacrificing stability or transient bounds.
Future work includes using arccm's in adaptive MPC, as well as developing new methods for computing contraction metrics in high-dimensions.
\section{Appendix}
\label{sec:appendix}
A Lyapunov generalization of the result presented in \cref{sec:proofs} are derived here.
Note that generalizing comes at the expense of losing convexity of the constructive conditions.

\begin{definition}
    \label{def:aisc-clf}
    A continuously differentiable function $V_{\theta} : \mathbb{R}^n \times \mathbb{R}^n \times \mathbb{R}^p \rightarrow \mathbb{R}_+$ where $V_{\theta}(x,x_d) = 0 \iff x(t) = x_d(t)$ is an \emph{adaptive} eISC-clf if the following hold
    \begin{gather}
        k_1 |x-x_d|^a \leq V_{\theta}(x,x_d) \leq k_2 |x-x_d|^a \tag{2} \\
        |x(t) - x_d(t)| \geq \sigma \left( \sup_{\tau \in [t_0,t]} | {\tilde{\theta}(\tau)} |\right) \implies  \nonumber \\
        \underset{u \in \mathbb{R}^m}{\mathrm{inf}} \Big\{ \nabla_x V_{\hat{\theta}}^\top F_{\hat{\theta}}(x,u) + \nabla_{x_d} V_{\hat{\theta}}^\top F_{\hat{\theta}}(x_d,u_d) - \nabla_x V_{\hat{\theta}}^\top \Delta^\top \tilde{\theta} \Big\} \nonumber \\
        \leq -k_3 |x-x_d|^a \label{eq:aisc-clf-1} \\
         -\mu V_\theta \leq \nabla_{\theta_i}V_\theta \leq \mu V_\theta, ~~ i=1,...,p \label{eq:aisc-clf-2}
    \end{gather}
    where $k_1,\,k_2,\,k_3,\,a,\,\mu\in \mathbb{R}_{>0}$, $\sigma:\mathbb{R}\rightarrow\mathbb{R}$ is a class $\mathcal{K}$ function, and $\hat{\theta} \in \Theta$ and $\tfrac{d}{dt}{\hat{\theta}} \in \Omega$ belong to compact sets.
    \label{def:lyap}
\end{definition}
\begin{theorem}
Let $\hat{\theta}$ be the current estimate of the model parameters $\theta$.
If an adaptive eISC-clf $\,V_{\hat{\theta}}(x,x_d)\,$ exists, then the closed-loop system is eISC with respect to the estimation error $\tilde{\theta}$. 
Hence, if $\tilde{\theta}\rightarrow 0$ then $x(t) \rightarrow x_d(t)$ exponentially.
\label{th:lyap}
\end{theorem}
\begin{proof}
Let $V_{\hat{\theta}}$ be an aISC-clf.
Consider when $|x-x_d| \geq \sigma ( \sup_{\tau \in [t_0,t]} |{\tilde{\theta}(\tau)}| )$.
Differentiating along \cref{eq:system} yields
\begin{equation*}
    \dot{V}_{\hat \theta} \leq -k_3|x-x_d|^a + \nabla_{\hat \theta} V_{\hat \theta}^\top \dot{{\hat \theta}}  \leq -\frac{k_3}{k_2}V_{\hat{\theta}} + p\mu V_{\hat \theta} \bigl|\dot{\hat \theta}\bigr|
\end{equation*}
where we make use of \cref{eq:isc-clf-pd,eq:aisc-clf-1,eq:aisc-clf-2} to obtain the second inequality.
If $\tfrac{d}{dt}\hat{\theta}$ is designed so $\rho(t) = \tfrac{k_3}{k_2} - p\mu |\tfrac{d}{dt}{\hat{\theta}}| > 0$ and since $\tilde \theta$ and $\tfrac{d}{dt}{\hat \theta}$ are bounded then letting $\bar{\rho} \triangleq \tfrac{k_3}{k_2}-p\mu ||\tfrac{d}{dt}{\hat \theta}||_\infty \leq \rho(t)$ yields $\dot{V}_{\hat{\theta}} \leq -\bar\rho \, V_{\hat \theta} \implies |x(t)-x_d(t)|  \leq \sigma_1(|x(0)-x_d(0)|) e^{-\lambda t}$ via \cref{eq:isc-clf-pd}.

Since this holds only when the implication \cref{eq:aisc-clf-1} is true, we can conclude $|x(t)-x_d(t)|  \leq \sigma_1(|x(0)-x_d(0)|) e^{-\lambda t} + \sigma_1(\sup_{\tau \in [t_0,t]} |{\tilde{\theta}(\tau)}|)$.
Therefore, the closed-loop system is eISC with respect to $\tilde{\theta}$.
Moreover, if $\tilde{\theta} \rightarrow 0$ then $x(t)\rightarrow x_d(t)$ exponentially. \qedhere
\end{proof}
We see that imposing a gradient bound on $V_{\hat{\theta}}(x,x_d)$ in \cref{eq:aisc-clf-2} similar to \cref{eq:arccm_c3} again bypasses requiring ISC with respect to the parameter estimation transients.
When $x(t) \neq x_d(t)$ then \cref{eq:aisc-clf-2} can also be interpreted as a bound on the log gradient of $V_{\hat{\theta}}(x,x_d)$.
Despite the similarity, it is well-known that computing clf's for general nonlinear system presents several difficulties that contraction theory circumvents \cite{manchester2017control}.

\FloatBarrier
\bibliographystyle{IEEEtran}
\bibliography{references}

\end{document}